\documentclass[draft,envcountsame,a4paper]{llncs}
\usepackage{amsmath}
\usepackage{amsfonts}
\usepackage{amssymb}

\newcommand\comment[1]{}
\newcommand\cb{Co-B\"uchi\ }

\DeclareMathOperator{\nex}{Next}
\DeclareMathOperator{\succs}{Succ}

\newcommand{\start}{\mathsf{start}}
\newcommand{\init}{\mathsf{init}}
\newcommand{\inpu}{\mathsf{input}}
\newcommand{\fail}{\mathsf{fail}}
\newcommand{\x}{\mathcal{X}}

\newcommand{\until}{\ \mathcal{U}\ }
\newcommand{\nats}{\mathbb N}

\begin{document}

\title{Synthesis of Succinct Systems}
\author{John Fearnley\inst{1} \and Doron Peled\inst{2} \and Sven Schewe\inst{1}}

\institute{Department of Computer Science, University of Liverpool, Liverpool,
UK \and	Department of Computer Science, Bar Ilan University, Ramat Gan 52900,
Israel }

\date{}
\maketitle

\begin{abstract}
Synthesis of correct by design systems from specification has recently attracted
much attention. The theoretical results imply that this problem is highly
intractable, e.g., synthesizing a system is 2EXPTIME-complete for an LTL
specification and EXPTIME-complete for CTL. An
argument in favor of synthesis is that the temporal specification is
highly compact, and the complexity reflects the large size of the system
constructed. A careful observation reveals that the size of the system is
presented in such arguments as the size of its state space. This view is a bit
biased, in the sense that the state space can be exponentially larger than the
size of a reasonable implementation such as a circuit or a program. Although
this alternative measure of the size of the synthesized system is more intuitive
(e.g., this is the standard way model checking problems are measured), research
on synthesis has so far stayed with measuring the system in terms of the
explicit state space. This raises the question of whether or not there always
exists a small system. In this paper, we show that this is the case if, and only
if, PSPACE = EXPTIME.
\end{abstract}

\section{Introduction}

Reactive synthesis is a research direction inspired by Church's problem
\cite{Church/63/Logic}. It focuses on systems that receive a constant stream of
inputs from an environment, and must, for each input, produce some output.
Specifically, we are given a logical specification which dictates how the system
must react to the inputs, and we must construct a
system that satisfies the specification for all possible inputs that the
environment could provide.

While the verification~\cite{Holzmann/97/SPIN,DBLP:conf/cav/CimattiCGGPRST02}
(the validation or refutation of the correctness of such a system) has gained
many algorithmic solutions and various successful tools, the synthesis
problem~\cite{ehl11tacas,DBLP:conf/cav/FiliotJR09,DBLP:conf/vmcai/PitermanPS06}
has had fewer results. One of the main problems is the complexity of the
synthesis problem. A classical result by Pnueli and Rosner~\cite{PR89} shows
that synthesis of a system from an LTL specification is 2EXPTIME-complete. It
was later shown by Kupferman and Vardi that synthesis for CTL specifications is
EXPTIME-complete~\cite{KV97}. A counter argument against the claim that
synthesis has prohibitive high complexity is that the size of the system
produced by the synthesis procedure is typically large. Some concrete
examples~\cite{Kupferman+Vardi/95/Formula} show that the size of the system
synthesized may need to be doubly exponentially larger than the LTL
specification. This, in fact, shows that LTL specifications are a very compact
representation of a system, rather than simply a formalism that is intrinsically
hard for synthesis. 

As we are interested in the relationship between the specification and the
synthesized system, a question arises with respect to the nature of the system
representation.  The classical synthesis problem regards the system as a
\emph{transition system} with an \emph{explicit} state space, and the size of
this system is the number of states and transitions. This is, to some extent, a
biased measure, as other system representations, such as programs or circuits
with memory, often have a much more concise representation: it is often possible
to produce a circuit or program that is exponentially smaller than the
corresponding transition system. For example, it is easy to produce a small
program that implements an $n$ bit binary counter, but a corresponding
transition system requires~$2^n$ distinct states to implement the same counter.
Thus, we ask the question of {\em what is the size of the minimal system
representation in terms of the specification?}

We look at specifications given in CTL, LTL, or as an
automaton, and study the relative synthesized system complexity. We
choose to represent our systems as online Turing machines with a bounded storage
tape. This is because there exist straightforward translations between online
Turing machines and the natural representations of a system, such as programs
and circuits, with comparable representation size. The online Turing machine
uses a read-only input tape to read the next input, a
write-only output tape to write the corresponding output, and
its storage tape to serve as the memory required to compute the corresponding
output for the current input.

The binary-counter example mentioned above showed that there are instances in
which an online Turing machine model of the specification is exponentially
smaller than a transition system model of that formula. In this paper we ask: is
this always the case? More precisely, for every 
CTL formula $\phi$, does there always exist an online Turing machine
$\mathcal{M}$ that models $\phi$, where the amount of space required to describe
$\mathcal{M}$ is polynomial in $\phi$? We call machines with this property
\emph{small}. Our answer to this problem is the following:
\begin{quote}
Every 
CTL formula has a small 
online Turing machine model (or no model at all) if, and only if, PSPACE =
EXPTIME.
\end{quote}

This result can be read in two ways. One point of view is that, since PSPACE is
widely believed to be a proper subset of EXPTIME, the ``if'' direction of our
result implies that it is unlikely that every CTL formula has a small online
Turing machine model. However, there is an opposing point of view. It is widely
believed that finding a proof that PSPACE $\ne$ EXPTIME is an extremely
difficult problem. The ``only if'' direction of our result implies that, if we
can find a family of CTL formulas that provably requires super-polynomial sized
online Turing machine models, then we have provided a proof that PSPACE $\ne$
EXPTIME. If it is difficult to find a proof that PSPACE $\ne$ EXPTIME, then it
must also be difficult to find such CTL formulas. This indicates that most CTL
formulas, particularly those that are likely to arise in practice, may indeed
have small online Turing machine models. 

Using an online Turing machine raises the issue of the time needed to respond to
an input. In principle, a polynomially-sized online Turing machine can take
exponential time to respond to each input. A small model may, therefore, take
exponential time in the size of the CTL formula to produce each output. This
leads to the second question that we address in this paper: for CTL, does there
always exist a small online Turing machine model that is \emph{fast}? The model
is fast if it always responds to each input in polynomial time. Again, our
result is to link this question to an open problem in complexity theory:
\begin{quote}
Every CTL formula 
has a small and fast (or no) online Turing machine model if, and only if, EXPTIME $\subseteq$ P/poly.
\end{quote}
P/poly is the class of problems solvable by a polynomial-time Turing machine
with an advice function that provides advice strings of polynomial size. It has
been shown that if EXPTIME $\subseteq$ P/poly, then the polynomial time
hierarchy collapses to the second level, and that EXPTIME itself
would be contained in the polynomial time hierarchy~\cite{KL82}.

Once again, this result can be read in two ways. Since many people believe that
the polynomial hierarchy is strict, the ``if'' direction of our result implies
that it is unlikely that all CTL formulas have small and fast models. On the
other hand, the ``only if'' direction of the proof implies that finding a family
of CTL formulas that do not have small and fast online Turing machine models is
as hard as proving that EXPTIME is not contained in P/poly. As before, if
finding a proof that EXPTIME $\nsubseteq$ P/poly is a difficult problem, then
finding CTL formulas that do not have small and fast models must also be a
difficult problem. This indicates that the CTL formulas that arise in practice
may indeed have small and fast models.

We also replicate these results for specifications given by \cb automata, which
then allows us to give results for LTL specifications.

\section{Preliminaries}

\subsection{CTL Formulas}

Given a finite set $\Pi$ of atomic propositions, the syntax of a CTL formula is
defined as follows:
\begin{align*}
\phi &::= p \; | \; \lnot \phi \; | \; \phi \lor \phi \; | \; A \psi \; | \; E
\psi, \\
\psi &::= \x \phi \; | \; \phi  \until  \phi,
\end{align*}
where $p \in \Pi$. For each CTL formula $\phi$ we define $|\phi|$ to give the
size of the parse tree for that formula.

Let $T = (V, E)$ be an infinite directed tree, with all edges pointing away from
the root. Let $l : V \rightarrow 2^\Pi$ be a labelling function. The semantics
of CTL are defined as follows. For each $v \in V$ we have:
\begin{itemize}
\item $v \models p$ if and only if $p \in l(v)$.
\item $v \models \lnot \phi$ if and only if $v \not\models \phi$.
\item $v \models \phi \lor \psi$ if and only if either $v \models \phi$ or $v
\models \psi$.
\item $v \models A \psi$ if and only if for all paths $\pi$ starting at $v$ we
have $\pi \models \psi$.
\item $v \models E \psi$ if and only if there exists a path $\pi$ starting at
$v$ with $\pi \models \psi$.
\end{itemize}
Let $\pi = v_1, v_2, \dots$ be an infinite path in $T$. We have:
\begin{itemize}
\item $\pi \models \x \psi$ if and only if $v_2 \models \psi$.
\item $\pi \models \phi  \until  \psi$ if and only if there exists $i \in
\nats$ such that $v_i \models \psi$ and for all $j$ in the range $1 \le j < i$
 we have $v_j \models \phi$.
\end{itemize}
The pair $(T, l)$, where $T$ is a tree and $l$ is a labelling function, is a
model of $\phi$ if and only if $r \models \phi$, where $r \in V$ is the root of
the tree. If $(T, l)$ is a model of $\phi$, then we write $T, l \models \phi$.

\subsection{Mealy Machines}

The synthesis problem is to construct a model that satisfies the given
specification. We will give two possible formulations for a model.
Traditionally, the synthesis problem asks us to construct a model in the form of
a transition system. We will represent these transition systems as \emph{Mealy
machines}, which we will define in this section. In a later section we will give
online Turing machines as an alternative, and potentially more succinct, model
of a specification.

A Mealy machine is a tuple $\mathcal{T} = (S, \Sigma_I, \Sigma_O, \tau, l,
\start, \inpu)$. The set~$S$ is a finite set of states, and the state $\start
\in S$ is the starting state. The set $\Sigma_I$ gives an input alphabet, and
the set $\Sigma_O$ gives an output alphabet. The transition function $\tau : S
\times \Sigma_I \rightarrow S$  gives, for each state and input letter, an
outgoing transition. The function $l: S \times \Sigma_I \rightarrow \Sigma_O$ is
a labelling function, which assigns an output letter for each transition. The
letter $\inpu \in \Sigma_I$ gives an initial input letter for the machine.

Suppose that~$\phi$ is a CTL formula that uses $\Pi_I$ as a set \emph{input}
propositions, and $\Pi_O$ as a set of \emph{output} propositions. Let
$\mathcal{T} = (S, 2^{\Pi_I}, 2^{\Pi_O}, \tau, l, \start, \inpu)$ be a Mealy
machine that uses sets of these propositions as input and output alphabets. A
sequence of states $\pi = s_0, s_1, s_2, \dots$ is an infinite path in
$\mathcal{T}$ if $s_0 = \start$, and if, for each $i$, there is a letter
$\sigma_i \in \Sigma_I$ such that $\tau(s_i, \sigma_i) = s_{i+1}$. We define
$\omega_i = \sigma_i \cup l(s_{i+1})$, where we take $\sigma_0 = \inpu$, to be
the set of input and output propositions at position~$i$ in the path. Then, for
each infinite path~$\pi$, we define the word $\sigma(\pi) = \omega_0, \omega_1,
\omega_2 \dots$ to give the sequence of inputs and outputs along
the path $\pi$. Furthermore, let $(T, l)$ be the infinite tree corresponding to
the set of words $\sigma(\pi)$, over all possible infinite paths $\pi$. We say
that $\mathcal{T}$ is a model of~$\phi$ if $T, l \models \phi$. Given a CTL
formula~$\phi$ and a Mealy machine $\mathcal{T}$, the CTL model checking problem
is to decide whether $\mathcal{T}$ is a model of $\phi$.

\begin{theorem}[\cite{Kupferman+Vardi+Wolper/00/CTL}]
\label{thm:ctllogspace}
Given a Mealy machine $\mathcal{T}$ and an CTL-formula $\phi$, the CTL model
checking problem can be solved in space polynomial in $|\phi| \cdot
\log|\mathcal{T}|$.
\end{theorem}

Given a CTL formula $\phi$, the \emph{CTL synthesis problem} is to decide
whether there exists a Mealy machine that is a model of~$\phi$. This problem is
known to be EXPTIME-complete.

\begin{theorem}[\cite{KV97}] 
\label{thm:realctl}
The CTL synthesis problem is EXPTIME-complete.
\end{theorem}

\subsection{Tree Automata}

Universal \cb tree automata will play a fundamental role in the proofs given in
subsequent sections, because we will translate each CTL formula~$\phi$
into a universal \cb tree automaton $\mathcal{U(\phi)}$. The automaton will
accept Mealy machines, and the language of the tree automaton will be
exactly the set of models accepted by~$\phi$. We will then use these automata to
obtain our main results.

A \emph{universal \cb tree automaton} is $\mathcal{A} = ( S, \Sigma_I, \Sigma_O
, \start, \delta, F)$, where~$S$ denotes a finite set of states, $\Sigma_I$ is a
finite input alphabet, $\Sigma_O$ is a finite output alphabet, $\start \in S$ is
an initial state, $\delta$ is a transition function, $F\subseteq S$ is a set of
final states. The transition function $\delta: S \times \Sigma_O \rightarrow
2^{S \times \Sigma_I}$ maps a state and an output letter to a set of pairs,
where each pair consists of a successor state and an input letter.

The automaton accepts Mealy machines that use $\Sigma_I$ and
$\Sigma_O$ as their input and output alphabets, and the acceptance mechanism is
defined in terms of run graphs. We define a \emph{run graph} of a universal \cb tree
automaton $\mathcal{A} = (S_\mathcal{A}, \Sigma_I, \Sigma_O ,
\start_\mathcal{A}, \delta, F_\mathcal{A})$ on a Mealy machine
$\mathcal{T} = (S_\mathcal{T}, \Sigma_I, \Sigma_O,
\tau, l_\mathcal{T}, \start_\mathcal{T})$ to be a minimal directed graph $G =
(V,E)$ that satisfies the following constraints:
\begin{itemize}
\item The vertices of $G$ satisfy $V \subseteq S_\mathcal{A} \times
S_\mathcal{T}$. 
\item The pair of initial states $(\start_\mathcal{A},\start_\mathcal{T})$
is contained in $V$.
\item Suppose that for a vertex $(q, t) \in V$, we have that $(q', \sigma_I) \in
\delta(q, l_\mathcal{T}(\sigma_I, t))$ for some input letter $\sigma_I$. An edge
from $(q, t)$ to $(q', \tau(t, \sigma_I))$ must be contained in $E$.
\end{itemize}
A run graph is \emph{accepting} if every infinite path $v_1, v_2, v_3, \dots \in
V^\omega$ contains only finitely many states in~$F_\mathcal{A}$. A Mealy machine
$\mathcal{T}$ is accepted by $\mathcal{A}$ if it has an accepting run graph. The
set of Mealy machines accepted by $\mathcal{A}$ is called its \emph{language},
and is denoted by $\mathcal{L}(\mathcal{A})$. The automaton is empty if, and
only if, its language is empty.

A universal \cb tree automaton is called a \emph{safety} tree automaton if $F =
\emptyset$. Therefore, for safety automata, we have that every run graph is
accepting, and we drop the $F=\emptyset$ from the tuple defining the automaton.
A universal \cb tree automaton is \emph{deterministic} if $|\delta(s,
\sigma_O))| = 1$, for all states~$s$, and output letters~$\sigma_O$.


\subsection{Online Turing Machines}

We use \emph{online} Turing machines as a formalisation of a concise model. An
online Turing machine has three tapes: an infinite input tape, an infinite
output tape, and a storage tape of bounded size. The input tape is read only,
and the output tape is write only. Each time that a symbol is read from the
input tape, the machine may spend time performing computation on the storage
tape, before eventually writing a symbol to the output tape. 

We can now define the synthesis problem for online Turing machines. Let $\phi$
be a CTL formula defined using $\Pi_I$ and $\Pi_O$, as the sets of input, and
output, propositions, respectively. We consider online Turing machines that use
$2^{\Pi_I}$ as the input tape alphabet, and $2^{\Pi_O}$ as the output alphabet.
Online Turing machines are required, after receiving an input symbol, to produce
an output before the next input symbol can be read. Therefore, if we consider
the set of all possible input words that could be placed on the input tape, then
the set of possible outputs made by the online Turing machine forms a tree. If
this tree is a model of $\phi$, then we say that the online Turing machine is a
model of $\phi$.

Given a CTL formula $\phi$, we say that an online Turing machine $\mathcal{M}$
is a \emph{small} model of $\phi$ if:
\begin{itemize}
\item the storage tape of $\mathcal{M}$ has length polynomial in $|\phi|$, and
\item the discrete control (i.e.\ the action table) of $\mathcal{M}$ has size
polynomial in $|\phi|$.
\end{itemize}
Note that a small online Turing machine may take an exponential number of
steps to produce an output for a given input. We say that an online Turing
machine is a \emph{fast} model of $\phi$ if, for all inputs, it always responds
to each input in time polynomial in $|\phi|$.

\section{Small Models Imply PSPACE = EXPTIME}
\label{sec:small2pe}

Let~$\phi$ be a CTL formula that has a model. In this section we show that, if
there is always a small online Turing machine that models~$\phi$, then PSPACE =
EXPTIME. Our approach is to guess a polynomially sized online Turing machine
$\mathcal{M}$, and then to use model checking to verify whether $\mathcal{M}$ is
a model of $\phi$. Since our assumption guarantees that we only need to guess
polynomially sized online Turing machines, this gives a NPSPACE = PSPACE
algorithm for solving the CTL synthesis problem. Our proof then follows from the
fact that CTL synthesis is EXPTIME-complete.

To begin, we show how model checking can be applied to an online Turing machine.
To do this, we first unravel the online Turing machine to a Mealy machine.

\begin{lemma}
\label{lem:tm}
For each CTL formula $\phi$, and each online Turing machine $\mathcal{M}$ that
is a model of $\phi$, there exists a Mealy machine $\mathcal{T}(\mathcal{M})$
such that $\mathcal{T}(\mathcal{M})$ is a model of $\phi$.
\end{lemma}

The size of $\mathcal{T(M)}$ will be exponential in the size of
$\mathcal{M}$, because the number of storage tape configurations of
$\mathcal{M}$ grows exponentially with the length of the tape. However, this is
not a problem because there exists a deterministic Turing machine that outputs
$\mathcal{T(M)}$, while using only $O(|\mathcal{M}|)$ space.

\begin{lemma}
\label{lem:logspace}
There is a deterministic Turing machine that outputs $\mathcal{T(M)}$, while
using $O(|\mathcal{M}|)$ space. 
\end{lemma}

Since the model checking procedure given in Theorem~\ref{thm:ctllogspace} uses
poly-logarithmic space, when it is applied to $\mathcal{T(M)}$, it will use
space polynomial in $|\mathcal{M}|$. Now, using standard techniques to compose
space bounded Turing machines (see~\cite[Proposition 8.2]{Pap94}, for example),
we can compose the deterministic Turing machine given by
Lemma~\ref{lem:logspace} with the model checking procedure given in
Theorem~\ref{thm:ctllogspace} to produce a deterministic Turing machine that
uses polynomial space in $|\mathcal{M}|$. Hence, we have shown that each online
Turing machine $\mathcal{M}$ can be model checked against~$\phi$ in space
polynomial in $|\mathcal{M}|$. Since Theorem~\ref{thm:realctl} implies that CTL
synthesis is EXPTIME-complete, we have the following theorem.

\begin{theorem}
\label{thm:easy}
Let $\phi$ be a satisfiable CTL formula. If there always exists an online Turing
machine $\mathcal{M}$ that models $\phi$, where $|\mathcal{M}|$ is polynomial in $\phi$, then PSPACE = EXPTIME.
\end{theorem}

\section{PSPACE = EXPTIME Implies Small Models}
\label{sec:hardctl}


In this section we show the opposite direction of the result given in
Section~\ref{sec:small2pe}. We show that if PSPACE = EXPTIME, then, for every CTL
formula~$\phi$ that has a model, there exists a polynomially
sized online Turing machine that is a model of~$\phi$. We start our proof of
this result with a translation from CTL to universal \cb tree automata.
In~\cite{KV97} it was shown that every CTL formula $\phi$ can be translated to
an \emph{alternating} \cb tree automaton $\mathcal{A(\phi)}$, whose language is the models of $\phi$.
It is relatively straightforward to translate this alternating tree automaton
$\mathcal{A(\phi)}$ into a universal tree automaton $\mathcal{U(\phi)}$. 

One complication is that the output alphabet of $\mathcal{U(\phi)}$ is not
$2^{\Pi_O}$. This is because the reduction to universal tree automata augments
each output letter with additional information, which is used by
$\mathcal{U(\phi)}$ to resolve the nondeterministic choices made by
$\mathcal{A(\phi)}$. Hence, each output letter of $\mathcal{U(\phi)}$ contains
an actual output $\sigma_O \in 2^{\Pi_O}$, along with some extra information.


Let $\mathcal{T} = (S, \Sigma_I, \Sigma_O, \tau, l_\mathcal{T}, \start, \inpu)$
be a Mealy machine, where each output $\sigma_O \in \Sigma_O$ contains some
element $a \in 2^{\Pi_O}$ with $a \subseteq \sigma_O$. We define $\mathcal{T}
\restriction \Pi_O$ to be a modified version of $\mathcal{T}$ that only produces
outputs from the set $2^{\Pi_O}$. Formally, we define $\mathcal{T}
\restriction \Pi_O$ to be the Mealy machine $\mathcal{T'} = (S, \Sigma_I,
2^{\Pi_O}, \tau, l_\mathcal{T'}, \start, \inpu)$ where, if $l_\mathcal{T}(s,
\sigma_I) = \sigma_O$, then we define $l_\mathcal{T'}(s, \sigma_I) = \sigma_O
\cap 2^{\Pi_O}$ for all $s \in S$. We have the following:

\begin{lemma}
\label{lem:ctl2ucbta}
Let~$\phi$ be a CTL formula, which is defined over the set $\Pi_I$ of input
propositions, and the set $\Pi_O$ of output propositions. We can construct a
universal \cb tree automaton $\mathcal{U(\phi)} = (S, \Sigma_I, \Sigma_O ,
\start, \delta, F, \inpu)$ such that:
\begin{itemize}
\item For every model $\mathcal{T} \in \mathcal{L(U(\phi))}$, we have that
$\mathcal{T} \restriction 2^{\Pi_O}$ is a model of $\phi$.
\item For every model $\mathcal{T}'$ of $\phi$ there is a model $\mathcal{T} \in \mathcal{L(U(\phi))}$
with $\mathcal{T} \restriction 2^{\Pi_O}=\mathcal T'$.
\item The size of the set $S$ is polynomial in $|\phi|$.
\item Each letter in $\Sigma_I$ and $\Sigma_O$ can be stored in space
polynomial in $|\phi|$. 
\item The transition function $\delta$ can be computed in time polynomial in
$|\phi|$.
\item The state $\start$ can be computed in polynomial time.
\end{itemize}
\end{lemma}

The techniques used in~\cite{SF07} show how the automaton given by
Lemma~\ref{lem:ctl2ucbta} can be translated into a
safety tree automaton $\mathcal{F(\phi)}$ such that the two automata are
emptiness equivalent.

\begin{lemma}[\cite{SF07}]
\label{lem:ucbta2dsta}
Given the universal \cb tree automaton $\mathcal{U(\phi)}$, whose state space is
$S_\mathcal{U}$, we can construct a deterministic safety tree automaton $\mathcal{F(\phi)}
=  (S, \Sigma_I, \Sigma_O , \start, \delta, \inpu)$ such that:
\begin{itemize}
\item If $\mathcal{L(U(\phi))}$ is not empty, then $\mathcal{L(F(\phi))}$ is not
empty. Moreover, if $\mathcal{T}$ is in $\mathcal{L(F(\phi))}$, then
$\mathcal{T}\restriction 2^{\Pi_O}$ is a model of $\phi$.
\item Each state in~$S$ can be stored in space polynomial in $|S_\mathcal{U}|$. 
\item The transition function~$\delta$ can be computed in time polynomial in
$|S_\mathcal{U}|$.
\item Each letter in $\Sigma_I$ and $\Sigma_O$ can be stored in space
polynomial in $|S_\mathcal{U}|$.
\item The state $\start$ can be computed in time polynomial in
$|S_\mathcal{U}|$.
\end{itemize}
\end{lemma}

We will use the safety automaton $\mathcal{F(\phi)}$ given by
Lemma~\ref{lem:ucbta2dsta} to construct a polynomially sized model of $\phi$.
This may seem counter intuitive, because the number of states in
$\mathcal{F(\phi)}$ may be exponential in~$\phi$. However, we do not need to
build $\mathcal{F(\phi)}$. Instead our model will solve language emptiness
queries for $\mathcal{F(\phi)}$.

For each state $s \in S$ in $\mathcal{F(\phi)}$, we define $\mathcal{F}^s(\phi)$
to be the automaton $\mathcal{F(\phi)}$ with starting state~$s$. The
\emph{emptiness problem} takes a CTL formula~$\phi$ and a state of
$\mathcal{F(\phi)}$, and requires us to decide whether
$\mathcal{L}(\mathcal{F}^s(\phi)) = \emptyset$. Note that the input has
polynomial size in $|\phi|$. We first argue that this problem can be solved in
exponential time. To do this, we just construct $\mathcal{F}^s(\phi)$. Since
$\mathcal{F}^s(\phi)$ can have at most exponentially many states in $|\phi|$,
and the language emptiness problem for safety automata can be solved in
polynomial time, we have that our emptiness problem lies in EXPTIME.

\begin{lemma}
\label{lem:empty}
For every CTL formula $\phi$, and every state $s$ in $\mathcal{F(\phi)}$ we can
decide whether $\mathcal{L}(\mathcal{F}^s(\phi)) = \emptyset$ in exponential
time.
\end{lemma}

The algorithm that we construct for Lemma~\ref{lem:empty} uses exponential time
and exponential space. However, our key observation is that, under the
assumption that PSPACE = EXPTIME, Lemma~\ref{lem:empty} implies that there must
exist an algorithm for the emptiness problem that uses polynomial space. We will
use this fact to construct $\mathcal{M(\phi)}$, which is a polynomially sized
online Turing machine that models $\phi$.

Let~$\phi$ be a CTL formula that uses $\Pi_I$ and $\Pi_O$ as the set of input
and output propositions, and suppose that $\phi$ has a model. Furthermore,
suppose that $\mathcal{F(\phi)} = (S, \Sigma_I, \Sigma_O , \start, \delta, F,
\inpu)$. The machine $\mathcal{M(\phi)}$ always maintains a current state $s \in
S$, which is initially set so that $s = \start$. Lemma~\ref{lem:ucbta2dsta}
implies that $s$ can be stored in polynomial space, and that setting $s =
\start$ can be done in polynomial time, and hence polynomial space.

Every time that $\mathcal{M(\phi)}$ reads a new input letter $\sigma_I \in
2^{\Pi_I}$ from the
input tape, the following procedure is executed. The machine loops through each
possible output letter $\sigma_O \in \Sigma_O$ and checks whether there is a
pair $(s', \sigma'_I) \in \delta(s, \sigma_O)$ such that
$\mathcal{L}(\mathcal{F}^{s'}(\phi)) \ne \emptyset$. When an output symbol
$\sigma_O$ and state $s'$ with this property are found, then the machine outputs
$\sigma_O \cap 2^{\Pi_O}$, moves to the state $s'$, and reads the next input
letter.

The fact that a suitable pair $\sigma_O$ and $s'$ always exists can be proved by
a simple inductive argument, which starts with the fact that
$\mathcal{L}(\mathcal{F}^{\start}(\phi)) \ne \emptyset$, and uses the fact that
we always pick a successor that satisfies the non-emptiness check. Moreover, it
can be seen that $\mathcal{M(\phi)}$ is in fact simulating some Mealy machine
$\mathcal{T}$ that is contained in  $\mathcal{L(F(\phi))}$. Therefore, by
Lemma~\ref{lem:ctl2ucbta}, we have that $\mathcal{M(\phi)}$ is a model of
$\phi$.

The important part of our proof is that, if PSPACE = EXPTIME, then this
procedure can be performed in polynomial space. Since each letter in $\Sigma_O$
can be stored in polynomial space, we can iterate through all letters
in~$\Sigma_O$ while using only polynomial space. By Lemma~\ref{lem:empty}, the
check $\mathcal{L}(\mathcal{F}^{s'}(\phi)) \ne \emptyset$ can be performed in
exponential time, and hence, using our assumption that PSPACE = EXPTIME, there
must exist a polynomial space algorithm that performs this check. Therefore, we
have constructed an online Turing machine that uses polynomial space and
models~$\phi$. Thus, we have shown the following theorem.

\begin{theorem}
\label{theo:hardctl}
Let~$\phi$ be a CTL formula that has a model. If PSPACE =
EXPTIME then there is an online Turing machine $\mathcal{M}$ that models $\phi$,
where $|\mathcal{M}|$ is polynomial in $\phi$.
\end{theorem}

Theorem~\ref{theo:hardctl} is not constructive. However, if a polynomially sized
online Turing machine that models~$\phi$ exists, then we can always find it in
PSPACE by guessing the machine, and then model checking it.

\section{Small And Fast Models Imply EXPTIME $\subseteq$ P/poly}
\label{sec:easyppoly}

In this section we show that, if all satisfiable CTL formulas have a
polynomially sized model that responds to all inputs within polynomial time,
then EXPTIME $\subseteq$ P/poly.

Let $\mathcal{A}_b$ be an alternating Turing machine that with a tape of
length~$b$ that is written in binary. Since the machine is alternating, its
state space $Q$ must be split into $Q_\forall$, which is the set of universal
states, and $Q_\exists$, which is the set of existential states. The first step
of our proof is to construct a CTL formula~$\phi_b$, such that all models
of~$\phi_b$ are forced to simulate $\mathcal{A}_b$. The formula will use a set
of input propositions $\Pi_I$ such that $|\Pi_I| = b$. This therefore gives us
enough input propositions to encode a tape configuration of $\mathcal{A}_b$. The
set of output propositions will allow us to encode a configuration of
$\mathcal{A}_b$. More precisely, the output propositions use:
\begin{itemize}
\item $b$ propositions to encode the current contents of the tape,
\item $\log_2(b)$ propositions to encode the current position of the
tape head, and
\item $\log_2(Q)$ propositions to encode $q$, which is the current state of the machine.
\end{itemize}

Our goal is to simulate $\mathcal{A}_b$. The environment will perform the
following tasks in our simulation. In the first step, the environment will
provide an initial tape configuration for $\mathcal{A}_b$. In each subsequent
step, the environment will resolve the non-determinism, which means that it will
choose the specific existential or universal successor for the current
configuration. In response to these inputs, our CTL formula will require that
the model should
faithfully simulate~$\mathcal{A}_b$. That is, it should start from the specified
initial tape, and then always follow the existential and universal choices made
by the environment. It is not difficult to write down a CTL formula that
specifies these requirements.

In addition to the above requirements, we also want our model to predict whether
$\mathcal{A}_b$ halts. To achieve this we add the following output
propositions:
\begin{itemize}
\item Let $C = |Q| \cdot 2^b \cdot b$ be the total number of configurations
that~$\mathcal{A}_b$ can be in. We add $\log(C)$ propositions to encode a
counter $c$, which gives the number of steps that have been simulated.
\item We add a proposition~$h$, and we will require that $h$ correctly predicts
whether $\mathcal{A}_b$ halts from the current configuration.
\end{itemize}
The counter~$c$ can easily be enforced by a CTL formula. To implement $h$, we
will add the following constraints to our CTL formula:
\begin{itemize}
\item If $q$ is an accepting state, then~$h$ must be true.
\item If $c$ has reached its maximum value, then $h$ must be false.
\item If $q$ is non-accepting and $c$ has not reached its maximum value then:
\begin{itemize}
\item If $q$ is an existential state, then $h\leftrightarrow E\x h$.
\item If $q$ is a universal state, then $h\leftrightarrow A\x h$.
\end{itemize}
\end{itemize}
These conditions ensure that, whenever the machine is in an existential state,
there must be at least one successor state from which $\mathcal{A}_b$ halts, and
whenever the machine is in a universal state, $\mathcal{A}_b$ must halt from all
successors. We will use~$\phi_b$ to denote the CTL formula that we have
outlined.

Suppose that there is an online Turing machine~$\mathcal{M}$ model of~$\phi_b$
that is both small and fast. We argue that, if this assumption holds, then we
can construct a polynomial time Turing machine $\mathcal{T}$ that solves the
halting problem for $\mathcal{A}_b$. Suppose that we want to decide whether
$\mathcal{A}_b$ halts on the input word $I$. The machine $\mathcal{T}$ does the
following:
\begin{itemize}
\item It begins by giving $I$ to $\mathcal{M}$ as the first input letter.
\item It then proceeds by simulating $\mathcal{M}$ until the first output letter
is produced.
\item Finally, it reads the value of $h$ from the output letter, and then
outputs it as the answer to the halting problem for $\mathcal{A}_b$.
\end{itemize}
Since $\mathcal{M}$ is both small and fast, we have that $\mathcal{T}$ is a
polynomial time Turing machine. Thus, we have the following lemma.

\begin{lemma}
\label{lem:phib}
If $\phi_b$ has a model that is small and fast, then there is a polynomial-size
polynomial-time Turing machine that decides the halting problem for~$\mathcal{A}_b$.
\end{lemma}

We now use Lemma~\ref{lem:phib} to prove the main result of this section: if
every CTL formula has a small and fast model, then EXPTIME $\subseteq$ P/poly.
We will do this by showing that there is a P/poly algorithm for solving an
EXPTIME-hard problem.

We begin by defining our EXPTIME-hard problem. We define the problem
\texttt{HALT-IN-SPACE} as follows. Let $\mathcal{U}$ be a
universal\footnote[1]{here, the word ``universal'' means an alternating Turing
machine that is capable of simulating all alternating Turing machines.}
alternating Turing machine. We assume that $\mathcal{U}$ uses space polynomial
in the amount of space used by the machine that is simulated.
The inputs to our problem are:
\begin{itemize}
\item An input word $I$ for $\mathcal{U}$.
\item A sequence of blank symbols $B$, where $|B| =
\mathit{poly}(|I|)$.
\end{itemize}
Given these inputs, \texttt{HALT-IN-SPACE} requires us to decide whether
$\mathcal{U}$ halts when it is restricted to use a tape of size $|I| + |B|$.
Since $B$ can only ever add a polynomial amount of extra space, it is
apparent that this problem is APSPACE-hard, and therefore EXPTIME-hard.

\begin{lemma}
\label{lem:his}
\texttt{HALT-IN-SPACE} is EXPTIME-hard.
\end{lemma}

A P/poly algorithm consists of two parts: a polynomial-time Turing
machine~$\mathcal{T}$, and a polynomially-sized \emph{advice function}~$f$. The
advice function maps the length of the input of $\mathcal{T}$ to a
polynomially-sized advice string. At the start of its computation, the
polynomial-time Turing machine~$\mathcal{T}$ is permitted to read $f(i)$,
where~$i$ is the length of input, and use resulting advice string to aid in its
computation. A problem lies in P/poly if there exists a machine~$\mathcal{T}$
and advice function~$f$ that decide that problem.

We now provide a P/poly algorithm for \texttt{HALT-IN-SPACE}. We begin by
defining the advice function~$f$. Let $\mathcal{U}_b$ be the machine
$\mathcal{U}$, when it is restricted to only use the first~$b$ symbols on its
tape. By Lemma~\ref{lem:phib}, there exists a polynomial-size polynomial-time
deterministic Turing machine $\mathcal{T}_b$ that solves the halting problem for
$\mathcal{U}_b$. We define $f(b)$ to give $\mathcal{T}_b$. Since $\mathcal{T}_b$
can be described in polynomial space, the advice function $f$ gives polynomial
size advice strings.

The second step is to give a polynomial-time algorithm that uses~$f$ to solve
\texttt{HALT-IN-SPACE}. The algorithm begins by obtaining $\mathcal{T}_{i+b} =
f(|I| + |B|)$ from the advice function. It then simulates $\mathcal{T}_{i+b}$ on
the input word $I$, and outputs the answer computed by $\mathcal{T}_{i+b}$. By
construction, this algorithm takes polynomial time, and correctly solves
\texttt{HALT-IN-SPACE}. Therefore, we have shown that an EXPTIME-hard problem
lies in P/poly, which gives the main theorem of this section.

\begin{theorem}
If every satisfiable CTL formula $\phi$ has a polynomial size model that
responds to all inputs in polynomial time, then EXPTIME
$\subseteq$ P/poly.
\end{theorem}

\section{EXPTIME $\subseteq$ P/poly Implies Small and Fast Models}
\label{sec:hardppoly}

Let~$\phi$ be a CTL formula that has a model. In this section
we show that if EXPTIME $\subseteq$ P/poly, then there always exists an
polynomially sized online Turing machine that is a model of~$\phi$, that also
responds to every input within a polynomial number of steps. 

The proof of this result closely follows the proof given in
Section~\ref{sec:hardctl}. In that proof, we looped through every possible
output letter and solved an instance of the emptiness problem. This was
sufficient, because we can loop through every output letter in polynomial space.
However, when we wish to construct a fast model, this approach does not work,
because looping through all possible output letters can take exponential time.

For this reason, we introduce a slightly modified version of the emptiness
problem, which we call the \emph{successor emptiness problem}. Using this
problem will allow us to find the correct output letter using binary search,
rather than exhaustive enumeration. The inputs to our problem will be a CTL
formula $\phi$, a state~$s$, an input letter~$\sigma_I \in 2^{\Pi_I}$ of
$\mathcal{F(\phi)}$, an integer $n$, and a bit string $w$ of length $n$. Given
these inputs, the problem is to determine whether there is a letter $\sigma_O
\in \Sigma_O$ such that:
\begin{itemize}
\item the first $n$ bits of $\sigma_O$ are $w$, and
\item there exists $(s', \sigma_I) \in \delta(s, \sigma_O)$ such that
$\mathcal{L}(\mathcal{F}^{s'}(\phi)) \ne \emptyset$.
\end{itemize}



Lemma~\ref{lem:ucbta2dsta} implies that the input size of this problem is
polynomial in $|\phi|$. In fact, if the CTL formula $\phi$ is fixed, then
Lemma~\ref{lem:ucbta2dsta} implies that all other input parameters have bounded
size. For a fixed formula $\phi$, let $(\phi, s, \sigma_I, n, w)$ be the input
of the successor emptiness problem that requires the longest representation. We
pad the representation of all other inputs so that they have the same length as
$(\phi, s, \sigma_I, n, w)$.

Next, we show how our assumption that EXPTIME $\subseteq$ P/poly allows us to
argue that successor emptiness problem lies in P/poly. Note that the successor
emptiness problem can be solved in exponential time, simply by looping through
all possible letters in $\sigma_O \in \Sigma_O$, checking whether the first $n$
bits of $\sigma_O$ are $w$, and then applying the algorithm of
Lemma~\ref{lem:empty}. Also note that our padding of inputs does not affect this
complexity. Therefore, the successor emptiness problem lies in EXPTIME, and our
assumption then implies that it also lies in P/poly.

Let $\mathcal{T}$ and $f$ be the polynomial time Turing machine and advice
function that witness the inclusion in P/poly. Our padding ensures that we have
that we have, for each CTL formula~$\phi$, a unique advice string in~$f$ that is
used by $\mathcal{T}$ to solve all successor emptiness problems for $\phi$. Hence,
if we append this advice string to the storage tape of $\mathcal{T}$, then we
can construct a polynomial time Turing machine (with no advice function) that
solves all instances of the successor emptiness problem that depend on~$\phi$.
Therefore, we have shown the following lemma.

\begin{lemma}
\label{lem:ppolyempty}
If EXPTIME $\subseteq$ P/poly then, for each CTL formula~$\phi$, there is a
polynomial time Turing machine that solves all instances of the successor
emptiness problem that involve $\phi$.
\end{lemma}

The construction of an online Turing machine that models $\phi$ is then the same
as the one that was provided in Section~\ref{sec:hardctl}, except that we use the
polynomial time Turing machine from Lemma~\ref{lem:ppolyempty} to solve the
successor emptiness problem in each step. More precisely, we use binary search
to find the appropriate output letter~$\sigma_O$ in each step. Since the size of~$\sigma_O$ is polynomial in~$|\phi|$, this can obviously be achieved in
polynomial time. Moreover, our online Turing machine still obviously uses only
polynomial space. Thus, we have established the main result of this section.

\begin{theorem}
Let~$\phi$ be a CTL formula that has a model. If EXPTIME
$\subseteq$ P/poly then there is a polynomially sized online Turing machine
$\mathcal{M}$ that models $\phi$ that responds to every input after a polynomial
number of steps.
\end{theorem}

\section{LTL Specifications}

In this section we extend our results to LTL. Since LTL specifications can be
translated into universal \cb tree automata~\cite{Kupferman+Vardi/05/Safraless},
our approach is to first extend our results so that they apply directly to
universal \cb tree automata, and then to use this intermediate step to obtain
our final results for LTL.

We start with universal \cb tree automata. Recall that the arguments in Sections
\ref{sec:hardctl} and \ref{sec:hardppoly} start with a CTL formula, translate
the formula into a universal \cb tree automaton, and then provide proofs that
deal only with the resulting automaton. Reusing the proofs from these sections
immediately gives the following two properties.


\begin{theorem}
\label{theo:ucb1}
Let $\mathcal U$ be a universal \cb automaton tree that accepts Mealy machines.
\begin{enumerate}
 \item If PSPACE = EXPTIME then there is an online Turing machine $\mathcal{M}$
in the language of $\mathcal U$, where $|\mathcal{M}|$ is polynomial in the
states and a representation of the transition function of $\mathcal U$.
\item If EXPTIME $\subseteq$ P/poly then there is a polynomially-sized online
Turing machine $\mathcal{M}$ that is accepted by $\mathcal U$, which responds to
every input after a polynomial number of steps.
\end{enumerate}
\end{theorem}

The other two directions can be proved using slight alterations of our existing
techniques. An analogue of the result in Section~\ref{sec:small2pe} can be
obtained by using the fact that checking whether online Turing machine $\mathcal M$ is
accepted by a universal \cb tree automaton $\mathcal U$ can be done in in
$O\big((\log |\mathcal U|+ \log |\mathcal{T(M)}|)^2\big)$ time~\cite[Theorem
3.2]{vardi1986automata}. For the result in Section~\ref{sec:easyppoly}, we can
obtain an analogue by using a very similar proof. The key difference is in
Lemma~\ref{lem:phib}, where we must show that there is universal \cb tree
automaton with the same properties as $\phi_b$. The construction of a suitable
universal \cb tree automaton appears in the full version of the paper. Thus, we
obtain the other two directions.


\begin{theorem}
\label{theo:ucb2}
Let $\mathcal U$ be a universal safety word automaton that accepts Mealy
machines. 
\begin{enumerate}
\item If there is always an online Turing machine $\mathcal{M}$ in the language of $\mathcal U$,
where $|\mathcal{M}|$ has polynomial size in the states of $\mathcal U$, then PSPACE=EXPTIME.
\item If there is always an online Turing machine $\mathcal{M}$ in the language
of $\mathcal U$, where $|\mathcal{M}|$ has polynomial size in the states of
$\mathcal U$, which also responds to every input after polynomially many steps,
then EXPTIME $\subseteq$ P/poly.
\end{enumerate}
\end{theorem}

Now we move on to consider LTL specifications. For LTL, the situation is more
complicated, because the translation from LTL formulas to universal \cb tree
automata does not give the properties used in Lemma~\ref{lem:ctl2ucbta}.


\begin{theorem}
\cite{Kupferman+Vardi/05/Safraless}
\label{theo:ltl2uct}
Given an LTL formula $\phi$, we can construct a universal
\cb tree automaton $\mathcal U_\phi$ with $2^{O(|\phi|)}$
states that accepts a Mealy machine ${\mathcal T}$ if, and only if, ${\mathcal
T}$ is a model of $\phi$.
\end{theorem}


Note that this translation gives a universal \cb tree automaton that has
exponentially many states in $|\phi|$. Unfortunately, this leads to less clean
results for LTL. For the results in
Sections~\ref{sec:small2pe} and~\ref{sec:hardctl}, we have the following
analogues. 


%

\begin{theorem}
\label{theo:ltl}
Let $\phi$ be an LTL specification.
\begin{enumerate}
\item If there is always an online Turing machine $\mathcal{M}$ exponential in the length of $\phi$ that models $\phi$, then EXPSPACE=2EXPTIME.
\item If PSPACE = EXPTIME, then there is a exponentially sized online Turing
machine $\mathcal{M}$ that models $\phi$.
\end{enumerate}
\end{theorem}

The first of these two claims is easy to prove: we can guess an exponentially
sized online Turing machine, and then model check it. For the second claim, we
simply apply Theorem~\ref{theo:ucb1}. 

In fact, we can prove a stronger statement for the second part of
Theorem~\ref{theo:ltl}. QPSPACE is the set of problems that can be solved in
$O(2^{(\log n_d)^c})$ space for some constant~$c$. We claim that if EXPTIME
$\subseteq$ QPSPACE, then, for every LTL formula $\phi$, there is an
exponentially sized online Turing machine $\mathcal{M}$ that models $\phi$. This
is because, in the proofs given in Section~\ref{sec:hardctl}, if we have an
algorithm that solves the emptiness problem in QPSPACE, then the online Turing
machine that we construct will still run in exponential time in the formula.

We can also prove one of the two results regarding small and fast online Turing
machines. The following result is implied by Theorem~\ref{theo:ucb1}.
\begin{theorem}
If EXPTIME $\subseteq$ P/poly then every LTL formula has an exponentially sized
online Turing machine model, which responds to every input after an exponential
number of steps.
\end{theorem}

We can strengthen this result to ``If all EXPTIME problems are polylogspace
reducible to P/poly then every LTL formula has an exponentially sized online
Turing machine model, which responds to every input after an exponential number
of steps'' with the same reason we used for strengthening the previous theorem:
in the proofs given in Section~\ref{sec:hardctl}, if we have an algorithm that
solves the emptiness problem in QPTIME using an advice tape of quasi-polynomial
size, then the online Turing machine that we construct will still run in
exponential time in the formula.

However, we cannot prove the opposite direction. This is because the proof used
in Theorem~\ref{theo:ucb2} would now produce an exponential time Turing
machine with an advice function that gives exponentially sized advice strings.
Therefore, we cannot draw any conclusions about the class P/poly.

\bibliographystyle{abbrv}
\bibliography{references}

\appendix
\newpage
\section{Online Turing Machines}
\label{app:tm}

In this section we give a full definition for Online Turing machines. These
definitions were not necessary for the main body of the paper, but they will be
used in some of the proofs in the appendix.

We first define Turing machines with input and output, and we then
introduce additional restrictions to ensure that each input is followed by
exactly one output. Once this has been done, we will formally define the
synthesis problem for Online Turing machines.

\subsection{Space Bounded Turing Machines with Input and Output}

A deterministic space bounded Turing machine with input and output is a
three-tape Turing machine defined by a tuple $(S, \Sigma_I, \Sigma_T, \Sigma_O,
\delta, \start, c, \init, \inpu)$. The set~$S$ is a finite set of states, and
the state $\start \in S$ is a starting state. The sets $\Sigma_I$, $\Sigma_T$,
and $\Sigma_O$ give the \emph{alphabet symbols} for the input, storage, and
output tapes. We require that there is a \emph{blank} symbol $\sqcup$, such that
$\sqcup$ is contained in $\Sigma_I$, $\Sigma_T$, and $\Sigma_O$. The function
$\delta$ is a \emph{transition function} which maps elements of $S \times
\Sigma_I \times \Sigma_T \times \Sigma_O$ to elements of $S \times (\Sigma_I
\times D) \times (\Sigma_T \times D) \times (\Sigma_O \times D)$, where $D =
\{\leftarrow, -, \rightarrow\}$ is the set of \emph{directions}. The number $c
\in \nats$ gives the \emph{space bound} for the machine, and the sequence $\init
\in (\Sigma_T)^c$ gives the initial contents of the storage tape, and the letter
$\inpu$ gives the initial input symbol. We define the size $|\mathcal{M}|$ of a
space bounded Turing machine $\mathcal{M}$ to be amount of space used by the
tuple $(S, \Sigma_I, \Sigma_T, \Sigma_O, \delta, \start, c, \init)$.

The machine has three tapes $I = I_1, I_2, \dots$, $T = T_1, T_2, \dots T_c$,
and $O = O_1, O_2, \dots$, which we call the \emph{input}, \emph{storage}, and
\emph{output} tapes, respectively. Note that, while the input and output tapes
are infinite, the storage tape contains exactly $c$ positions. For all $i \in
\nats$ we have $I_i \in \Sigma_I$, and $O_i \in \Sigma_O$. For all $i$ in the
range $1 \le i \le c$ we have that $T_i \in \Sigma_T$. The tapes are initialized
as follows: the input tape $I$ contains an infinitely long \emph{input word},
where the first letter $I_1 = \inpu$. The output tape $O$ contains an infinite
sequences of blank symbols, and the storage tape $T$ contains the initial
storage word $\init$.

A \emph{position} gives the current state of the machine, along with the
position of the three tape heads. Formally, a position is a tuple of the form
$(s, i, j, k)$, where $s \in S$, $i,k \in \nats$, and $j$ is in the range $1 \le
j \le c$. For each $i \in \nats$, and direction $d \in D$,  we define:
\begin{equation*}
\nex(i, d) = \begin{cases}
i - 1 & \text{if $d = \leftarrow$ and $j > 0$,} \\
i + 1 & \text{if $d = \rightarrow$,} \\
i & \text{otherwise.} 
\end{cases}
\end{equation*}
We also define:
\begin{equation*}
\nex^c(i, d) = \begin{cases}
i - 1 & \text{if $d = \leftarrow$ and $j > 0$,} \\
i + 1 & \text{if $d = \rightarrow$ and $j < c$,} \\
i & \text{otherwise.} 
\end{cases}
\end{equation*}

The machine begins in the position $(\start, 0, 0, 0)$. We now describe one step
of the machine. Suppose that the machine is in position $(s, i, j, k)$, and that
$\delta(s, I_i, T_j, O_k) = (s', (\sigma_I, d_1), (\sigma_T, d_2), (\sigma_O,
d_3))$. First the symbols $\sigma_I$, $\sigma_T$, and $\sigma_O$ are written to
$I_i$, $T_j$, and $O_k$, respectively. Then, the machine moves to the position
$(s', \nex(i, d_1), \nex^c(j, d_2), \nex(k, d_3))$, and the process repeats. 

\subsection{Online Turing Machines}

An \emph{online} Turing machine is a Turing machine with input and output that
has additional restrictions on the transition function $\delta$. We wish to
ensure the following property: the machine may only read the symbol at position
$i$ on the input tape after it has written a symbol to position $i-1$ on the
output tape. Moreover, once a symbol has been written to the output tape, we
require that it can never be changed. Thus, the machine must determine the
first~$i$ symbols of the output before the $i+1$th symbol of the input can be
read.

To this end, we partition the set $S$ into the set~$S_I$ of \emph{input states},
and the set~$S_O$ of \emph{output states}, and we require that $\start \in S_O$.
While the machine is in an output state, it is prohibited from moving the input
tape head, or from moving the output tape head left. Furthermore, the machine
moves from an output state to an input state only when a symbol is written to
the output tape. Similarly, when the machine is in an input state it is
prohibited from moving the output tape head, and the machine only moves from an
input state to an output state when the input tape head is moved right. Finally,
the input tape is read-only. This means that in every state, the machine is
prohibited from overwriting the symbols on the input tape.

Formally, let $(s, i, j, k)$ be a position, and suppose that $\delta(s, I_i,
T_j, O_k) = (s', (\sigma_I, d_1), (\sigma_T, d_2), (\sigma_O, d_3))$. If $s \in
S_I$, then we require:
\begin{itemize}
\item  The direction $d_1$ is either $-$ or $\rightarrow$, and the direction $d_3$ is $-$.
\item  The symbol $\sigma_I = I_i$, and the symbol $\sigma_O = O_k$. 
\item If $d_1 = -$ then we require that $s' \in S_I$, and if $d_1 = \rightarrow$
then $s' \in S_O$.
\end{itemize}
Similarly, if $s \in S_O$, then we require:
\begin{itemize}
\item The direction $d_1$ is $-$, and the direction $d_3$ is either $-$ or $\rightarrow$. 
\item The symbol $\sigma_1 = I_i$, and if $d_3 = -$ then $\sigma_3
= O_k$.
\item If $d_3 = -$ then $s' \in S_O$, and if $d_3 = \rightarrow$ then $s' \in
S_I$.
\end{itemize}

\subsection{The Synthesis Problem}

We are interested in online Turing machines with input alphabet $\Sigma_I =
2^{\Pi_I}$ and output alphabet $\Sigma_O = 2^{\Pi_O}$, where $\Pi_I$  is a set
of input variables, and $\Pi_O$ is a set of output variables. The sets $\Pi_I$
and $\Pi_O$ are required to be disjoint. We will use $\emptyset$ as the blank
symbol for the input and output tapes.

Let $\mathcal{M}$ be an online Turing machine. For each input word~$\sigma$
that can be placed on the input tape, the machine produces an output word on the
output tape. This output word is either an infinite sequence of outputs made by
the machine, or a finite sequence of outputs, followed by an infinite sequence
of blanks. The second case arises when the machine runs forever while producing
only a finite number of outputs. 

We define $\mathcal{M}(\sigma)$ to be the combination of the inputs given to the
machine on the input variables, and the outputs made by the machine on the
output variables. Formally, let $I = I_0, I_1, \dots$, and $O =
O_0, O_1, \dots$ be the contents of the input and output tapes after
the machine has been allowed to run for an infinite number of steps on the
input word~$\sigma$. We define $\mathcal{M}(\sigma) = \sigma_0, \sigma_1,
\dots$, where $\sigma_i = I_i \cup O_i$.

Let $\phi$ be an LTL formula that uses $\Pi_I \cup \Pi_O$ as the set of atomic
propositions. We say that an online Turing machine $\mathcal{M}$
is a model of $\phi$ if $\mathcal{M}(\sigma) \models \phi$ for all input
words~$\sigma$. We say that $\phi$ is \emph{realizable} if there exists an
online Turing machine $\mathcal{M}$ that satisfies $\phi$. The LTL synthesis
problem for the formula $\phi$ is to decide whether $\phi$ is realizable.

On the other hand, let $\phi$ be a CTL formula that uses $\Pi_I \cup \Pi_O$ as
the set of atomic propositions. Note that, since an online Turing machine cannot
read the $i$th input letter until it has produced $i-1$ outputs, if $\sigma$ and
$\sigma'$ are two input words that agree on the first $i$ letters, then
$\mathcal{M}(\sigma)$ and $\mathcal{M}(\sigma')$ must agree on the first $i$
letters. Note also that, since the initial input letter is fixed, all of these
words must agree on the first letter. Hence, the set of words
$\{\mathcal{M}(\sigma) \; : \; \sigma \in (\Sigma_I)^{\omega}\}$ must form an
infinite directed labelled tree $(T, l)$. We say that $\mathcal{M}$ is a model
of $\phi$ if $T, l \models \phi$.

\section{Proof of Lemma~\ref{lem:ctl2ucbta}}

To prove this Lemma, we first invoke the result of~\cite{KV97} to argue that CTL
formulas can be translated to \emph{alternating} \cb tree automata, and then
argue that these automata can be translated into universal \cb tree automata.
Therefore, we proceed by first giving definitions for alternating tree automata,
and then providing a proof for Lemma~\ref{lem:ctl2ucbta}.

\subsection{Alternating Tree Automata}

We now define alternating \cb tree automata. An \emph{alternating \cb tree
automaton} is a tuple $\mathcal{A} = ( S, \Sigma_I, \Sigma_O , \start, \delta,
F, \inpu)$, where $S$ denotes a finite set of states, $\Sigma_I$ is a finite
input alphabet, $\Sigma_O$ is a finite output alphabet, $\start \in S$ is an
initial state, $\delta$ is a transition function, $F\subseteq S$ is a set of
final states, and $\inpu$ is an initial input letter. The transition function
$\delta: S \times \Sigma_O \rightarrow \mathbb{B}^+(S \times
\Sigma_I)$ maps a state and an output letter to a boolean formula that is built
from elements of $S \times \Sigma_I$, conjunction~$\wedge$, disjunction~$\vee$,
$\mathit{true}$, and $\mathit{false}$. Universal \cb tree automata correspond to
alternating \cb tree automata in which all formulas given by $\delta$ are
conjunctions.

The automaton runs on Mealy machines that use $\Sigma_I$ and $\Sigma_O$ as their
input and output alphabets. The acceptance mechanism is defined in terms of run
graphs. We define the \emph{run graph} of an alternating \cb tree automaton $ \mathcal{A} =
(S_\mathcal{A}, \Sigma_I, \Sigma_O , \start_\mathcal{A}, \delta_\mathcal{A},
F_\mathcal{A}, \inpu)$ with respect to a Mealy machine $\mathcal{T} = (S_\mathcal{T},
\Sigma_I, \Sigma_O, \tau, l_\mathcal{T}, \start_\mathcal{T}, \inpu_\mathcal{T})$
to be a (minimal) directed graph $G = (V,E)$ that satisfies
the following constraints:
\begin{itemize}
\item The vertices of $G$ satisfy $V \subseteq S_\mathcal{A} \times
S_\mathcal{T}$. 
\item The pair of initial states $(\start_\mathcal{A},\start_\mathcal{T})$
is contained in $V$.
\item For each vertex $(q,t)\in V$, and each input letter $\sigma_I \in
\Sigma_I$, the set 
\begin{equation*}
\left\{ \big(q',\sigma_I\big) \in S \times \Sigma_I \mid
\Big(\big(q,t\big),\big(q',\tau(t,\sigma_I)\big)\Big) \in E \right\}
\end{equation*} 
is a satisfying assignment of $\delta(q,l_\mathcal{T}(\sigma_I, t))$.
\end{itemize}
A run graph is \emph{accepting} if every infinite path $v_1, v_2, v_3, \dots \in
V^\omega$ contains only finitely many final states. A Mealy machine
$\mathcal{T}$ is accepted by $\mathcal{A}$ if it has an accepting run graph. The
set of Mealy machines accepted by an automaton $\mathcal{A}$ is called its
\emph{language} $\mathcal{L}(\mathcal{A})$. An automaton is empty if, and only
if, its language is empty.

The acceptance of a Mealy machine can also be viewed as the outcome of a game,
where player \emph{accept} chooses, for a pair $(q,t) \in S_\mathcal{A} \times
S_\mathcal{T}$, a set of atoms that satisfies $\delta(q,l_\mathcal{T}(\sigma_I,
t))$. Player \emph{reject} then chooses one of these atoms, and then moves to
the corresponding state. The Mealy machine is accepted if, and only if, player
\emph{accept} has a strategy that ensures that all paths visit~$F$ only a finite
number of times.

\subsection{Translating CTL to Universal \cb tree automata}

The reason we are interested in alternating \cb tree automata is that there
exists a translation from CTL formulas to alternating \cb tree
automata~\cite{KV97}. One technical complication is that the automata presented
in~\cite{KV97} accept \emph{Moore machines} rather than Mealy machines. Moore
machines are similar to Mealy machines, except the labelling function labels
states rather than transitions.

Formally, a Moore machine is a tuple $\mathcal{T} = (S, \Sigma_I, \Sigma_O,
\delta, l, \start, \inpu)$. The set~$S$ is a finite set of states, and the state
$\start \in S$ is the starting state. The set $\Sigma_I$ gives an input
alphabet, and the set $\Sigma_O$ gives an output alphabet. The transition
function $\delta : S \times \Sigma_I \rightarrow S$  gives, for each state and
input letter, an outgoing transition. The function $l: S \rightarrow \Sigma_O$ is a
labelling function, which gives, for each state, a corresponding output letter.
The letter $\inpu \in \Sigma_I$ gives an initial input letter for the machine.

We have the following Lemma.

\begin{lemma}
\cite{KV97}
\label{lem:ctl2aa}
For every CTL formula $\phi$ there is an alternating \cb tree automaton
$\mathcal{A(\phi)} = (S, \Sigma_I, \sigma_O , \start, \delta, F, \inpu)$ such that:
\begin{itemize}
\item $\mathcal{L(A)}$ is the set of Moore machines that model $\phi$.
\item We have that $|S|$ and $|\delta|$ have size polynomial in $|\phi|$.

Essentially, $|S|$ is contains the temporal subformulas of $\phi$, which explains its size immediately.
To account for the small size of $\phi$, note that for the universal subformulas a proof obligations are sent to all directions, while, for existential temporal subformulas, proof obligations are sent into some directions.
We therefore can simply represent these two types of obligations by two symbols, e.g., $\forall$ and $\exists$ or $\square$ and $\Diamond$.
\item Each letter in~$\Sigma_I$ and~$\Sigma_O$ can be stored in space
polynomial in $|\phi|$.
\item The starting state can be computed in polynomial time.
\end{itemize}
\end{lemma}

We first show how to transform $\mathcal{A(\phi)}$ into an alternating \cb tree
automaton that accepts Mealy machines. To do this, we must ensure that
$\mathcal{A(\phi)}$ accepts only input preserving models.
Hence, we construct a second alternating \cb tree automaton
$\mathcal{A'(\phi)} = (S', \Sigma_I, \Sigma_O', \start',
\delta', F, \inpu)$. 
We expand the state space so that $S' = \Sigma_I \times S$, and we will require
that each state remembers the last input. Therefore, we set the starting state
to be $\start' = (\inpu, \start)$. Let $(\sigma_I, s) \in S'$ be a state, and
let $(\sigma'_I, \sigma'_O)$ be a pair of input and output letters. We also
expand the output alphabet so that $\Sigma'_O = \Sigma_I \times \Sigma_I$. That
is, each output letter also contains an input letter. Then, for each state
$(\sigma_I, s) \in S'$, and each output letter $(\sigma'_I, \sigma'_O) \in
\Sigma'_O$, we define:
\begin{equation*}
\delta'((\sigma_I, s), (\sigma'_I, \sigma'_O)) = \begin{cases}
\emptyset & \text{if $\sigma_I \ne \sigma_I'$,} \\
\{(\sigma''_I, s''), \sigma''_I) \; : \; (s'', \sigma''_I) \in \delta(s, \sigma'_O) & \text{otherwise.}
\end{cases}
\end{equation*}
Having made this transformation, we now have that $\mathcal{L(A'(\phi))}$ is
the set of Mealy machines that model $\phi$.

We now show that $\mathcal{A'(\phi)}$ can be translated into a universal \cb
automaton $\mathcal{U(\phi)} = (S, \Sigma_I, \Sigma''_O , \start,
\delta_\mathcal{U)}, F, \inpu)$. Let $\mathcal{T} \in \mathcal{L(A(\phi))}$ be a
Mealy machine that is accepted by $\mathcal{A'(\phi)}$, and let $G = (V, E)$ be
the run graph of $\mathcal{T}$ on $\mathcal{A'(\phi)}$. Note that, by the
definition of run graphs, for each state $(q, t) \in V$ we use exactly one
satisfying assignment to generate the outgoing edges from $(q, t)$. The idea
behind this proof is to use the output symbols of $\mathcal{T}$ to store this
satisfying assignment.

Formally, we define the extended alphabet $\Sigma''_O := \Sigma'_O \times (S
\rightarrow 2^{S \times \Sigma_I})$. Each output letter of the Mealy machine
contains an output letter $\sigma'_O \in \Sigma_O$ of $\mathcal{A'(\phi)}$,
along with $|S|$ lists of satisfying assignments. Since $S$ has size polynomial
in $\phi$, and each element of $\Sigma_I$ has size in $O(|\phi|)$, each element
of $\Sigma'_O$ can be stored in space polynomial in $O(|\phi|)$.

Let $q \in S$ be a state of $\mathcal{A'(\phi)}$, let $\sigma_O \in \Sigma'_O$ be an
output letter, and let $\gamma : S \rightarrow 2^{S \times \Sigma_I}$. If
$\gamma(q)$ is a satisfying assignment of $\delta(q, \sigma_O)$, then we add the
transition $\delta_\mathcal{U}(q, (\sigma_O, \gamma)) = \gamma(s)$. Although the
function $\delta_\mathcal{U}$ may require exponential space to describe fully,
we can compute in polynomial time, for a given $q$, $\sigma_O$, and $\gamma$,
whether $\delta_\mathcal{U}(q, (\sigma_O, \gamma))$ is a transition. This is
because $\delta$ has polynomial size in $|\phi|$, and checking whether
$\gamma(s)$ is a satisfying assignment can easily be done in polynomial time.

Note that, for each Mealy machine $\mathcal{T} \in \mathcal{L(U(\phi))}$, we can
use labels of each state to argue that there must be a corresponding run graph
of $\mathcal{T} \restriction 2^{\Pi_O}$ on $\mathcal{A'(\phi)}$. On the other
hand, if $\mathcal{T}$ has a run graph on $\mathcal{A'(\phi)}$, then we can
easily use the satisfying assignments used in this run graph to construct a
Mealy machine $\mathcal{T'} \in \mathcal{L(U(\phi))}$ with $\mathcal{T'}
\restriction 2^{\Pi_O} = \mathcal{T}$. Therefore, we have that there exists a
$\mathcal{T}' \in \mathcal{L(A'(\phi))}$ if and only if there exists a
$\mathcal{T} \in \mathcal{L(U(\phi))}$, and that $\mathcal{T} \restriction
2^{\Pi_O}$ is a model of $\phi$ if, and only if, $\mathcal{T} \in
\mathcal{L(U(\phi))}$. This completes the proof of Lemma~\ref{lem:ctl2ucbta}.

\section{Proof of Lemma~\ref{lem:tm}}

\begin{proof}
Suppose that $\mathcal{M} = (S, \Sigma_I, \Sigma_T, \Sigma_O, \delta, \start, c,
\init, \inpu)$, and let $(S_I, S_O)$ be the partition of~$S$ into the input and
output states. We will assume that $\Sigma_I = 2^{\Pi_I}$, and that $\Sigma_O =
2^{\Pi_O}$, for some sets $\Pi_I$ and $\Pi_O$ of input and output propositions. 

We define a Mealy machine $\mathcal{T}(\mathcal{M}) = (S_\mathcal{T}, \Sigma_I,
\Sigma_O, \tau_\mathcal{T}, l_\mathcal{T}, \start_\mathcal{T}, \inpu)$ as
follows. The state space is defined as follows. Firstly we have the normal
states $S_N = S \times \Sigma_I \times \nats_{\le c} \times \Sigma_O \times
(\Sigma_T)^c$, where $\nats_{\le c} = 1, 2, \dots, c$, which represent the
computational states that the Turing machine can be in. In addition to these, we
also have a special state $\fail$, which will be used to indicate that the
online Turing machine runs forever without writing an output symbol. Therefore,
we have $S_\mathcal{T} = S_N \cup \{\fail\}$.


All states $a \in S_N$ are tuples of the form $(s, \sigma_I, j, \sigma_O, T =
\langle T_1, T_2, \dots, T_c \rangle)$, where $s$ is a state of the Turing
machine, $j$ is the current position of the storage tape head, $\sigma_I$ is the
symbol at the head of the input tape, $\sigma_O$ is the last symbol written to
the output tape, and~$T$ is the current state of the storage tape. Since we know
that $I_1 = \inpu$ and $O_1 = \emptyset$ in the initial state of the machine,
the starting state of Mealy machine will be $\start_\mathcal{T}
= (\start, \inpu, 1, \emptyset, \init)$. 

We now define $\tau_\mathcal{T}$ and $l_\mathcal{T}$. Firstly, we define
$\tau(\fail, \sigma_I) = \fail$, and we define $l_\mathcal{T}(\fail, \sigma_I) =
\emptyset$, for all $\sigma_I \in \Sigma_I$. Now we consider the states $a = (s,
\sigma_I, j, \sigma_O, T = \langle T_1, T_2, \dots, T_c \rangle)$ in the set
$S_N$. Note that the definition of an online Turing machine ensures that there
is always a blank symbol at the head of the output tape. Therefore, suppose
that:
\begin{equation*}
\delta(s, \sigma_I, T_j, \emptyset) = (s', (\sigma'_I, d_1), (\sigma'_T, d_2),
(\sigma'_O, d_3)).
\end{equation*}
If~$a$ has $d_1 = \rightarrow$, then we say that~$a$ is an~\emph{input
state}, and if~$a$ has $d_3 = \rightarrow$, then we say that~$a$ is an
\emph{output state}. In our reduction, we will only consider the input states,
and the starting $\start_\mathcal{T}$.
For all other states $s$ we define $\tau_\mathcal{T}(s, \sigma_I) = \fail$ for
all input letters $\sigma_I \in \Sigma_I$.

Before we define $\tau_\mathcal{T}$, we first define a helper function $\succs$.
This function will be used to find the transitions. Let $a = (s, \sigma_I, j,
\sigma_O, T = \langle T_1, T_2, \dots, T_c \rangle)$ be a normal state, and let
$T'$ be the tape $T$ with the $j$th symbol replaced with $\sigma'_T$. If $d_3 =
-$ then we define $\succs(a)$ to be 
\begin{equation*}
a' = (s', \sigma_I, \nex(j, d_2), \sigma_O, T').
\end{equation*} 
On the other hand, if
$d_3 = \rightarrow$ then we define $\succs(a)$ to be 
\begin{equation*}
a' = (s', \sigma_I, \nex(j, d_2), \sigma'_O, T').
\end{equation*} 
Note that this definition correctly remembers the last symbol that was written
to the output tape.

For each input state $a = (s, \sigma_I, j, \sigma_O, T = \langle T_1, T_2,
\dots, T_c \rangle)$ and each input letter $\sigma_I$, let $a'$ be $\succs(a)$,
where the input letter is replaced by $\Sigma_I$. We define $\pi(a, \sigma_I)$
be the path that starts at $a'$ and follows $\succs(a)$ until another input
state is reached. Note that this path may be infinite if the Turing machine runs
forever without requesting another input. If $\pi(a, \sigma_I)$ ends at a state
$a''$, and if $\sigma_O$ is the output letter at $a''$, then we define $\tau(a,
\sigma_I) = a''$, and we define $l_\mathcal{T}(a, \sigma_I) = \sigma_O$.

On the other hand, if $\pi(a, \sigma_I)$ is an infinite path, then we have two
cases to consider. If $\pi(a, \sigma_I)$ never visits an output state, then we
define $\tau_\mathcal{T}(a, \sigma_I) = \fail$, and we define $l_\mathcal{T}(a,
\sigma_I) = \emptyset$. On the other hand, if $\pi(a, \sigma_I)$ does visit an
output state $a''$, and if $\sigma_O$ is letter output at $a''$, then we define
$\tau_\mathcal{T}(a, \sigma_I) = \fail$, and we define $l_\mathcal{T}(a,
\sigma_I) = \sigma_O$.

To see that this reduction is correct, note that, in an online Turing machine,
exactly one output is written to the output tape for each input that is read
from the input tape. Therefore, for every state~$a$, our transition function
$l_\mathcal{T}(a, \sigma_I)$ correctly moves to a state $a' = (s, j, \sigma_O,
T)$, where $\sigma_O$ is the output given by the online Turing machine for the
input $\sigma_I$. Moreover, if the Turing machine runs for an infinite number of
steps while producing only a finite number of outputs, then $\mathcal{T}$ will
correctly produce an infinite sequence of blank symbols, and it also correctly
outputs the final output symbol. Therefore, we have that $\mathcal{T(M)}$ is a
model of~$\phi$ if and only if~$\mathcal{M}$ is a model of~$\phi$. \qed
\end{proof}

\section{Proof of Lemma \ref{lem:logspace}}

\begin{proof}
Given the online Turing machine $\mathcal{M} = (S, \Sigma_I, \Sigma_T, \Sigma_O,
\delta, \start, c, \init)$, our task is to output the Mealy machine
$\mathcal{T}(\mathcal{M}) = (S_\mathcal{T}, \Sigma_\mathcal{T},
\tau_\mathcal{T}, l_\mathcal{T}, \start_\mathcal{T}, \inpu)$. Recall that each
normal state $s \in S_N$ is a tuple of the form $(s, j, \sigma_O, T )$.
Obviously the parameters~$s$, $j$, and $\sigma_O$, can be stored in
$O(\mathcal{M})$ space. Moreover, since the description of $\mathcal{M}$
contains $\start$, which is a tape of length~$c$, the tape $T$ can also be
stored in $O(|\mathcal{M}|)$ space. Since the state space of $S_\mathcal{T}$
consists of $S_N$ and $\fail$, we have that each state of
$\mathcal{T}(\mathcal{M})$ can be stored in $O(|\mathcal{M}|)$ space.

Our algorithm for outputting $\mathcal{T(M)}$ is as follows. We begin by
outputting $\start_\mathcal{T}$. Then we output the state $\fail$ along with the
outgoing transitions and labels from $\fail$. We then cycle through each normal
state $s \in S_N$, and output $s$ All of these operations can obviously be done
in $O(|\mathcal{M}|)$ space.

Finally, we must argue, for each normal state $s \in S_N$, that the outgoing
transitions and labels from $s$ can be computed in $O(|\mathcal{M}|)$ space. Let
$\sigma_I \in \Sigma_I$ be an input letter. Our algorithm iteratively follows
the function $\succs(s, \sigma_I)$ until we find an input state. If, while
iterating $\succs(s, \sigma_I)$, we encounter an output state $s'$, then we
remember the letter $\sigma_O$ that was outputted. If we find an input state
$s'$, then we output $\tau_\mathcal{T}(s, \sigma_I) = s'$ and $l_\mathcal{T}(s,
\sigma_I) = \sigma_O$.

On the other hand, we may never find an
input state. Therefore, we also maintain a counter, which counts the number of
times that $\succs$ has been followed. If the counter reaches $|S_N|$, then we
know that the online Turing machine must run forever without reading its next
input. In this case, we output $\tau_\mathcal{T}(s, \sigma_I) = \fail$. If
an output letter $\sigma_O$ has been remembered, then we output
$l_\mathcal{T}(s, \sigma_I) = \sigma_O$, otherwise we output $l_\mathcal{T}(s,
\sigma_I) = \emptyset$.

To implement this procedure, we must remember at most $2$ states, one for $s$,
and one for the current state. We also remember at most one output letter. We
must also maintain a counter that uses $\log(|S_N|)$ bits, and $|S_N| \in
2^{O(n)}$. Therefore, this procedure can be implemented in $O(|\mathcal{M}|)$
space. \qed
\end{proof}

\section{Proof of Theorem~\ref{thm:easy}}

\begin{proof}
We show that, under the assumption that there is always a model of size $c$, the
CTL synthesis problem can be solved in polynomial space. Since
Theorem~\ref{thm:realctl} implies that CTL synthesis is EXPTIME-complete, we
will therefore prove that PSPACE = EXPTIME.

The algorithm is as follows. We first non-deterministically guess an online
Turing machine $\mathcal{M}$ with with $|\mathcal{M}| \in
O(\text{poly}(|\phi|))$. Then we model check against the input formula $\phi$,
using the Turing machine given by Lemma~\ref{lem:logspace} and the Turing
machine given by Theorem~\ref{thm:ctllogspace}. Since the output of the first
Turing machine has size $2^{O(|\mathcal{M}|)}$, we have that the second Turing
machine uses $O(|\mathcal{M}|)$ space. Using standard techniques to compose
space bounded Turing machines (see~\cite[Proposition 8.2]{Pap94}, for example),
we obtain a Turing machine that solves the CTL synthesis problem in NPSPACE
= PSPACE. \qed
\end{proof}

\section{Proof of Lemma~\ref{lem:empty}}

\begin{proof}
From the formula $\phi$ we can construct $\mathcal{F}^s(\phi) = (S, \Sigma_I,
\Sigma_O , s, \delta, F)$ in exponential time by doing the following: first we
loop through each possible state in $S$ and output it. Since
Lemma~\ref{lem:ucbta2dsta} guarantees that each state can be stored in space
polynomial in~$\phi$, this procedure can take at most exponential time
in~$\phi$. Then, we loop through each $(s, \sigma_O) \in S \times \Sigma_O$ and
output the transition $\delta(s, \sigma_O)$. Again, since
Lemma~\ref{lem:ucbta2dsta} implies that each member of $S \times \Sigma_O$ can
be written in polynomial space, and therefore this procedure takes at most
exponential time.

So far we have shown that the states and transitions of $\mathcal{F}^s(\phi)$
can be constructed in exponential time, while using exponential space. We call a
state $s' \in S$ \emph{rejecting} if $\delta(s', \sigma_O) = \emptyset$ for all
output letters $\sigma_O \in \Sigma_O$. It is not difficult to see that 
$\mathcal{L}(\mathcal{F}^s(\phi)) = \emptyset$ if, and only if, all possible
paths from~$s$ lead to a rejecting state. Thus, we can solve the emptiness
problem by solving a simple reachability query on the automaton that we have
constructed. Since reachability can be solved in polynomial time, and the
description of our automaton uses exponential space, this reachability query can
be answered in exponential time. \qed
\end{proof}

\section{Proof of Lemma~\ref{lem:his}}

\begin{proof}
We prove this fact by reduction from the halting problem for an alternating
polynomial space Turing machine, which we will denote as \texttt{APSPACE-HALT}.
This is an APSPACE-complete problem. There are two inputs to
\texttt{APSPACE-HALT}:
\begin{itemize}
\item a polynomially-space bounded alternating Turing machine $\mathcal{T}$, and
\item an input word $I_\mathcal{T}$ for $\mathcal{T}$.
\end{itemize}
Since $\mathcal{T}$ is polynomially space bounded, we can assume that it comes
equipped with a function $s : \nats \rightarrow \nats$, where $s(i)$ gives the
total amount of space used for an input of length $i$. Since computing $s$ can
be done by evaluating a polynomial, we know that $s(i)$ can be computed in
polynomial time.

We now provide a polynomial time Turing-reduction to \texttt{HALT-IN-SPACE}.
Since $\mathcal{U}$ is a universal Turing machine, there must exist an input~$I$
for $\mathcal{U}$ such that $\mathcal{U}$ accepts~$I$ if and only if
$\mathcal{T}$ accepts $I_\mathcal{T}$. We fix~$I$ for the rest of the proof.
Since $\mathcal{U}$ uses polynomially more space than the machine that it
simulates, we can assume that it comes equipped with a function $t : \nats
\rightarrow \nats$, where $t(i)$ gives the amount of space used by $\mathcal{U}$
while simulating a machine that uses~$i$ space. Again, since computing $t(i)$
can be done by evaluating a polynomial, we have that $t(i)$ can be computed in
polynomial time. Therefore, to complete our reduction, we construct $B$ to be a
sequence of blanks of length $t(s(|I|))$, and then solve \texttt{HALT-IN-SPACE}
for $I$ and $B$. \qed
\end{proof}

\section{LTL and Automata}

\subsection{LTL Formulas}

Given a finite set $\Pi$ of atomic propositions, the syntax of an LTL formula is
defined as follows:
\begin{equation*}
\phi ::= p \; | \; \lnot \phi \; | \; \phi \lor \phi \; | \; \x \phi \; | \;
\phi \until  \phi,
\end{equation*}
where $p \in \Pi$. For each LTL formula, we define $|\phi|$ to give the size of
the formula, which is the size of the parse tree for that formula.

Let $\sigma = \sigma_0, \sigma_1, \do ts$ be an infinite word where each symbol
$\sigma_i \in \Pi$. For each $i \in \nats$, we define the semantics of an LTL
formula $\phi$ as follows:
\begin{itemize}
\item $\sigma, i \models p$ if and only if $p \in \sigma_i$.
\item $\sigma, i \models \lnot \phi$ if and only if $\sigma, i \not\models
\phi$.
\item $\sigma, i \models \phi \lor \psi$ if and only if either $\sigma, i
\models \phi$ or $\sigma, i \models \psi$.
\item $\sigma, i \models \x \phi$ if and only if $\sigma, i+1 \models \phi$.
\item $\sigma, i \models \phi \until \psi$ if and only if there exists $n
\ge i$ such that $\sigma, n \models \psi$ and for all $j$ in the range $i \le j
< n$ we have $\sigma, j \models \phi$.
\end{itemize}
A word $\sigma$ is a model of an LTL formula $\phi$ if and only if $\sigma, 0
\models \phi$. If $\sigma$ is a model of $\phi$, then we write $\sigma \models
\phi$.

Let~$\phi$ be an LTL formula that uses $\Pi_I \cup \Pi_O$ as a set of atomic
propositions. We say that a Mealy machine is a model of~$\phi$ if $\sigma(\pi)
\models \phi$ for every infinite path $\pi$ that begins at the starting state.
Given an LTL formula~$\phi$ and a Mealy machine $\mathcal{T}$, the LTL model
checking problem is to decide whether $\mathcal{T}$ is a model of $\phi$. 

\begin{theorem}[\cite{vardi1986automata}]
Given a Mealy machine $\mathcal{T}$ and an LTL-formula $\phi$,
the LTL model checking problem can be solved in $O((\log |\mathcal{T}| +
|\phi|)^2)$ space.
\end{theorem}

The LTL synthesis problem is defined in the same way as the CTL synthesis
problem: given an LTL formula $\phi$, we must decide whether there exists a
Mealy machine that is a model of $\phi$.

\begin{theorem}[\cite{PR89}]
\label{thm:realltl}
The LTL synthesis problem is 2EXPTIME-complete.
\end{theorem}

\subsection{Proofs for Theorem \ref{theo:ucb2}}

\begin{lemma}
Let $\mathcal U$ be a realisable universal safety word automaton.
If there always exists an online Turing machine $\mathcal{M}$ that realises $\mathcal U$ while taking only $O(\text{poly}(|\mathcal U|))$ time between reading two input letters, then EXPTIME $\subset$ P/poly.
\end{lemma}

\begin{proof}
We use a reduction from the halting problem of a universal space bounded alternating Turing machine.

We first define the set of output propositions $\Pi_O$ that will be used by our universal \cb automaton.
Our intention is again that each letter $\sigma_O \in \Sigma_O =
2^{\Pi_O}$ should encode a configuration of an alternating Turing machine with a
storage tape of length $b$. This the storage tape itself can be represented
using $b' = b \cdot \log_2 |\Sigma_T|$ atomic propositions $p_1,\cdots,p_{b'}$.
We also use $b$ atomic propositions $t_1,\ldots,t_b$ to encode the position of
the tape head: the propositions $t_i$ is true if and only if the tape head is at
position $i$ of the tape. We use $l = \log_2(|Q|)$ atomic propositions,
where~$Q$ is the set of states in our Turing machine, to encode $q$, which is
the current state in the configuration. We also use $l+b'+\log_2|b|$ atomic
propositions to encode a counter $c$, which will count the number of steps that
have been executed. Finally, and most importantly, we include one
proposition~$h$, and we will require that~$h$ accurately predicts whether the
alternating Turing machine will eventually halt from the current configuration.
Different to the reduction from CTL, we also have to include a way to resolve existential choices in the model.
We therefore also include atomic propositions that refer to the directions that serve as witnesses for the fact that $h$ is \emph{true} for existential states or \emph{false} for universal states.
We refer to this successor as the \emph{witness successor}.

We now specify the universal safety automaton $\mathcal U$.
Since $\mathcal A$ is an alternating Turing machine, the transition
function between configurations is not deterministic. Instead, in each step
there is either a universal or an existential choice that must be made. We
will allow the environment to resolve these decisions. Since $\Sigma_I$ contains
enough letters to encode every possible configuration of $\mathcal{A}$, there
are obviously more than enough letters in $\Sigma_I$ to perform this task.

To check the correctness of these transitions, the universal safety automaton would, for each transition, have a corresponding input letter.
It would send, for each cell of the tape of $\mathcal A$, for the finite control of $\mathcal A$, and for the position the read/write head should be in after the transition, a state to all successors.
This state would not only contain this first input symbol will be interpreted as the initial state of our alternating Turing machine
$\mathcal{A}$.
While it sends the obligations to all successors, they are only interpreted on the single successor where the input read (and hence represented in the label) is $\sigma$.
(Reading a different input leads to immediate acceptance.)

Once
the existential and universal decisions of $\mathcal A$ have been resolved, the formula requires the
model to output the next configuration of the alternating Turing machine. In
other words, the environment will pick a specific branch of the computation of
$\mathcal{A}$, and therefore all Mealy machines in the language of $\mathcal U$ must be capable of producing
all possible computation branches of $\mathcal{A}$. It is not difficult to
produce a universal safety automaton that encodes these requirements.

However, we have one final requirement that must be enforced: that, in the first step of the computation, the
proposition $h$ correctly predicts whether the current configuration eventually
halts. This can be achieved  by adding the following requirements to our universal safety automaton.
\begin{itemize}
\item If $q$ is an accepting state, then $h$ must be true.
\item If $q$ is non-accepting and $c$ has reached its maximum value, then $h$ must be false.
\item If $q$ is non-accepting and $c$ has not reached its maximum value then:
\begin{itemize}
\item If $q$ is an existential state and $h$ is true, then $h$ must be true for the witness successor.
\item If $q$ is an existential state and $h$ is false, then $h$ must be false for all successors.
\item If $q$ is an existential state and $h$ is true, then $h$ must be true for all successors.
\item If $q$ is an universal state and $h$ is false, then $h$ must be false for the witness successor.
\end{itemize}
\end{itemize}

Therefore, we have constructed a universal safety word automaton $\mathcal U$ such that, for every model
of $\mathcal U$, the first output from the model must solve the halting problem for a
$b$-space bounded alternating Turing machine. \qed 
\end{proof}

\begin{theorem}
\label{thm:easyppoly}
Let $\mathcal U$ be a realisable universal safety word automaton.
If there always exists an online Turing machine $\mathcal{M}$, with $|\mathcal{M}| \in O(\text{poly}(|\mathcal U|))$, that is accepted by $\mathcal U$, then PSPACE = EXPTIME.
\end{theorem}

\begin{proof}
A proof that the synthesis problem for these automata is EXPTIME complete is contained in the proof of the previous lemma.

Model checking if $\mathcal M$ is accepted by $\mathcal U$ can be done in space $O\big((\log |\mathcal U|+ \log |\mathcal{T(M)}|)^2\big)$, using the reduction from \cite{vardi1986automata} Theorem 3.2
(note that the language of $\mathcal U$ is the complement of the language of the same automaton read as a nondeterministic reachability automaton, where blocking translates to immediate acceptance and vice versa) to the emptiness problem of nondeterministic B\"uchi word automata \cite{DBLP:journals/tcs/SistlaVW87}.
\qed
\end{proof}

\end{document}